\documentclass[11pt]{article}
\usepackage{graphicx}
\usepackage{tabularx}
\usepackage{url}
\usepackage{amsthm}
\usepackage{amsmath}
\usepackage{amsfonts}
\usepackage[labelfont=bf]{caption}

\theoremstyle{plain}
\newtheorem{theorem}{Theorem}
\newtheorem{lemma}{Lemma}
\newtheorem{proposition}{Proposition}

\theoremstyle{definition}
\newtheorem{definition}{Definition}

\theoremstyle{remark}

\newtheorem{exmp}{Example}

\date{}

\begin{document}

\title{A Local Approach to Studying the Time and Space Complexity of Deterministic and Nondeterministic Decision Trees}

\author{Kerven Durdymyradov and Mikhail Moshkov \\
Computer, Electrical and Mathematical Sciences \& Engineering Division \\ and Computational Bioscience Research Center\\
King Abdullah University of Science and Technology (KAUST) \\
Thuwal 23955-6900, Saudi Arabia\\ \{kerven.durdymyradov,mikhail.moshkov\}@kaust.edu.sa
}

\maketitle

\begin{abstract}
Decision trees and decision rules are intensively studied and used in different areas of computer science. The questions important for the theory of decision trees and rules include relations between decision trees and decision rule systems, time-space tradeoff for decision trees, and  time-space tradeoff for decision rule systems.
In this paper, we study arbitrary infinite binary information systems each of which consists of an infinite set called universe and an infinite set of two-valued functions (attributes) defined on the universe. We consider the notion of a problem over information system, which is described by a finite number of attributes and a mapping associating a decision to each tuple of attribute values. As algorithms for problem solving, we investigate deterministic and nondeterministic decision trees
that use only attributes from the problem description.
Nondeterministic decision trees are representations of decision rule systems that sometimes have less space complexity than the original rule systems.
As time and space complexity, we study the depth and the number of nodes in the decision trees. In the worst case, with the growth of the number of attributes in the problem description, (i) the minimum depth of deterministic decision trees grows either  as a logarithm or linearly, (ii) the minimum depth of nondeterministic decision trees either is bounded from above by a constant or grows linearly, (iii) the minimum number of nodes in deterministic decision trees has either  polynomial or exponential growth, and (iv) the minimum number of nodes in nondeterministic decision trees has either  polynomial or exponential growth. Based on these results, we divide the set of all infinite binary information systems  into three complexity classes. This allows us to identify nontrivial relationships between deterministic decision trees and decision rules systems represented by  nondeterministic decision trees. For each class, we study issues related to time-space trade-off for deterministic and nondeterministic decision trees.
\end{abstract}


\section{Introduction}

\label{S1}

Decision trees and decision rules are intensively studied  and used in different areas of computer science. The questions important for the theory of decision trees and rules include relations between decision trees and decision rule systems, time-space tradeoff for decision trees, and  time-space tradeoff for decision rule systems.

In this paper, instead of decision rule systems we study nondeterministic decision trees. These trees can be considered as  representations of decision rule systems that sometimes have less space complexity than the original rule systems. We study  problems over  infinite binary information systems and divide the set of all infinite binary information systems  into three complexity classes depending on the worst case time and space complexity of deterministic and nondeterministic decision trees solving problems. This allows us to identify nontrivial relationships between deterministic decision trees and decision rule systems represented by nondeterministic decision trees. For each complexity class, we study issues related to time-space trade-off for deterministic and nondeterministic decision trees.

Decision trees \cite%
{AbouEisha19,Alsolami20,Breiman84,Moshkov05,Moshkov11,Rokach07} and systems
of decision rules \cite%
{BorosHIK97,BorosHIKMM00,ChikalovLLMNSZ13,FurnkranzGL12,MPZ08,Pawlak91,PawlakPS08,PawlakS07,SkowronR92}
are widely used as classifiers to predict a decision for a new object, as a
means of knowledge representation, and as algorithms for solving problems of
fault diagnosis, computational geometry, combinatorial optimization, etc.

Decision trees and rules are among the most interpretable models for
classifying and representing knowledge \cite{Molnar22}. In order to better
understand decision trees, we should not only minimize the number of their
nodes, but also the depth of decision trees to avoid the consideration of
long conjunctions of conditions corresponding to long paths in these trees.
Similarly, for decision rule systems, we should minimize both the total length of rules
and the maximum length of a rule in the system. When we consider decision
trees and decision rule systems as algorithms (usually sequential for
decision trees and parallel for decision rule systems), we should have in
mind the same bi-criteria optimization problems to minimize space and time
complexity of these algorithms.

In this paper, we represent systems of decision rules as nondeterministic
decision trees to compress them and to emphasize the possibility of
processing different decision rules in parallel. We consider deterministic
and nondeterministic decision trees as algorithms and study their space and
time complexity, paying particular attention to time and space complexity
relationships. Examples of deterministic and nondeterministic decision trees computing Boolean function $x_1 \wedge x_2$ can be found in Fig. \ref{fig1}.

\begin{figure}[ht]
	\begin{center}
		\includegraphics[width=80mm]{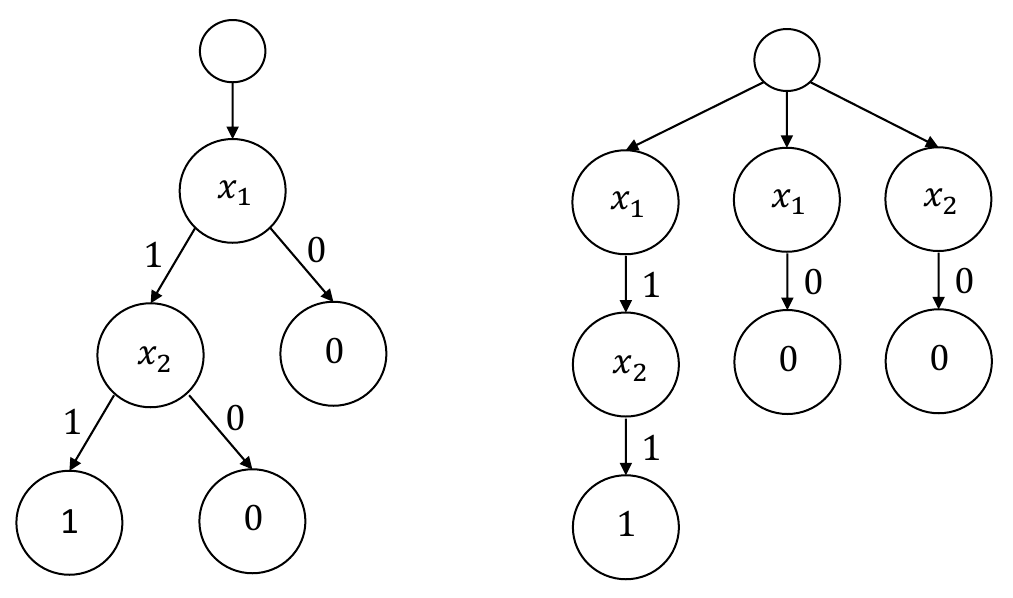}
		\caption{Deterministic and nondeterministic decision trees computing function $x_1 \wedge x_2$}
		\label{fig1}
	\end{center}
\end{figure}

Infinite systems of attributes and decision trees over these systems have
been intensively studied, especially systems of linear and algebraic
attributes and the corresponding linear \cite{Dobkin78,Dobkin79,Moravek72}
and algebraic decision trees \cite%
{Ben-Or83,Gabrielov,Grigoriev95,Grigoriev98,Steele82,Yao92,Yao94}. Years
ago, one of the authors initiated the study of decision trees over arbitrary
infinite systems of attributes \cite%
{Moshkov94b,Moshkov94a,Moshkov95,Moshkov96,Moshkov05a}. In this paper, we study
decision trees over arbitrary infinite systems of binary attributes
represented in the form of infinite binary information systems.

General information system introduced by Pawlak \cite{Pawlak81} consists of a universe (a set of
objects) and a set of attributes (functions with finite image) defined on
the universe. An information system is called infinite, if both its universe and the set of
attributes are infinite.  An information
system is called binary if each of its attributes has values from the set $%
\{0,1\}$.

Any problem over an information system is described by a finite number of
attributes that divide the universe into domains in which these attributes
have fixed values. A decision is attached to each domain. For a given object
from the universe, it is required to find the decision attached to the
domain containing this object.

As algorithms solving these problems, deterministic and nondeterministic
decision trees are studied. As time complexity of a decision tree, we
consider its depth, i.e., the maximum number of nodes labeled with
attributes in a path from the root to a terminal node. As space complexity
of a decision tree, we consider the number of its nodes.

There are two approaches to the study of infinite information systems: local
when in decision trees solving a problem we can use only attributes from the
problem description, and global when in the decision trees solving a problem
we can use arbitrary attributes from the considered information system.
In this paper, we study decision trees in the framework of the local approach.

To the best of our knowledge, time-space trade-offs for decision trees over
infinite information systems were not studied in the framework of the local
approach prior to the present paper except for its conference version \cite{Kerven23b}, which does not contain proofs.
The paper \cite{Moshkov23} was the first one in which
the time-space trade-offs for decision trees over infinite information
systems were studied in the framework of the global approach. 

Results obtained in \cite{Moshkov23} are  different from the results obtained in the present paper. Apart from the difference in approach, in \cite{Moshkov23}, the set of all infinite information systems is divided not into three but into five families and the criteria for the behavior of functions characterizing the minimum depth of the  deterministic and nondeterministic
decision tree are completely different in comparison with the present paper. However, many of the definitions and results of the two papers appear similar. In the present paper,
we use some auxiliary statements proved in \cite{Moshkov23} and adapt some proofs from \cite{Moshkov23} to the case of the local approach.

Based on the results obtained in the present paper and in \cite%
{Moshkov05,Moshkov11}, we describe possible types of behavior of four
functions $h_{U}^{ld},h_{U}^{la},L_{U}^{ld},L_{U}^{la}$ that characterize
worst case time and space complexity of deterministic and nondeterministic
decision trees over an infinite binary information system $U$ (index $l$
refers to the local approach). Decision trees solving a problem can use only attributes from the problem description.

The function $h_{U}^{ld}$ characterizes the growth in the worst case of the
minimum depth of a deterministic decision tree solving a problem with the growth of the number
of attributes in the problem description. The function $h_{U}^{ld}$ is
either grows as a logarithm or linearly.

The function $h_{U}^{la}$ characterizes the growth in the worst case of the
minimum depth of a nondeterministic decision tree solving a problem  with the growth of the
number of attributes in the problem description. The function $h_{U}^{la}$
is either bounded from above by a constant or grows linearly.

The function $L_{U}^{ld}$ characterizes the growth in the worst case of the
minimum number of nodes in a deterministic decision tree solving a problem with the growth of
the number of attributes in the problem description. The function $L_{U}^{ld}
$ has either polynomial or exponential growth.

The function $L_{U}^{la}$ characterizes the growth in the worst case of the
minimum number of nodes in a nondeterministic decision tree solving a
problem   with the
growth of the number of attributes in the problem description. The function $%
L_{U}^{la}$ has either polynomial or exponential growth.

Each of the functions $h_{U}^{ld},h_{U}^{la},L_{U}^{ld},L_{U}^{la}$ has two
types of behavior. The tuple $(h_{U}^{ld},h_{U}^{la},L_{U}^{ld},L_{U}^{la})$
has three types of behavior. All these types are described in the paper and
each type is illustrated by an example.

There are three complexity classes of infinite binary information systems
corresponding to the three possible types of the tuple $%
(h_{U}^{ld},h_{U}^{la},L_{U}^{ld},L_{U}^{la})$. For each class, we study
joint behavior of time and space complexity of decision trees. The obtained
results are related to time-space trade-off for deterministic and
nondeterministic decision trees.

A pair of functions $(\varphi ,\psi )$ is called a boundary $ld$-pair of the
information system $U$ if, for any problem over $U$, there exists a
deterministic decision tree over $z$, which solves this problem and for
which the depth is at most $\varphi (n)$ and the number of nodes is at most $%
\psi (n)$, where $n$ is the number of attributes in the problem description.
 An
information system $U$ is called $ld$-reachable if the pair $%
(h_{U}^{ld},L_{U}^{ld})$ is a boundary  $%
ld$-pair of the system $U$. For nondeterministic decision trees, the notions
of a boundary $la$-pair of an information system and $%
la$-reachable information system are defined in a similar way. For
deterministic decision trees, the best situation is when the considered
information system is $ld$-reachable: for any boundary $ld$-pair $(\varphi
,\psi )$ for an information system $U$ and any natural $n$, $\varphi (n)\geq
h_{U}^{ld}(n)$ and $\psi (n)\geq L_{U}^{ld}(n)$. For nondeterministic
decision trees, the best situation is when the information system is $la$%
-reachable.

For all complexity classes, all information systems from the class are $ld$%
-reachable. For two out of the three complexity classes, all information
systems from the class are $la$-reachable. For the remaining class, all
information systems from the class are not $la$-reachable. For all
information systems $U$ that are not $la$-reachable, we find nontrivial
boundary $la$-pairs, which are sufficiently close to $(h_{U}^{la},L_{U}^{la})
$.

The rest of the paper is organized as follows: Section \ref{S2} contains
main results, Sections \ref{S3}-\ref{S5} -- proofs of these results, and Section \ref{S6} -- short conclusions.

\section{Main Results}
\label{S2}

Let $A$ be an infinite set and $F$ be an infinite set of functions that are
defined on $A$ and have values from the set $\{0,1\}$. The pair $U=(A,F)$ is
called an \emph{infinite binary information system} \cite{Pawlak81}, the
elements of the set $A$ are called \emph{objects}, and the functions from $F$
are called \emph{attributes}. The set $A$ is called sometimes the \emph{%
universe} of the information system $U$.

A \emph{problem over} $U$ is a tuple of the form $z=(\nu ,f_{1},\ldots
,f_{n})$, where $\nu :\{0,1\}^{n}\rightarrow \mathbb{N}$, $\mathbb{N}$ is
the set of natural numbers $\{1,2,\ldots \}$, and $f_{1},\ldots ,f_{n}\in F$. We do not require attributes $f_{1},\ldots ,f_{n}$ to be pairwise distinct.
The problem $z$ consists in finding the value of the function $z(x)=\nu
(f_{1}(x),\ldots ,f_{n}(x))$ for a given object $a\in A$. The value $\dim
z=n$ is called the \emph{dimension of the problem} $z$.

Various problems of combinatorial optimization, pattern recognition, fault
diagnosis, probabilistic reasoning, computational geometry, etc., can be
represented in this form.

As algorithms for problem solving we consider decision trees. A \emph{%
decision tree over the information system} $U$ is a directed tree with a
root in which the root and edges leaving the root are not labeled, each
terminal node is labeled with a number from $\mathbb{N}$, each working node
(which is neither the root nor a terminal node) is labeled with an attribute
from $F$, and each edge leaving a working node is labeled with a number from
the set $\{0,1\}$. A decision tree is called \emph{deterministic} if only
one edge leaves the root and edges leaving an arbitrary working node are
labeled with different numbers.

Let $\Gamma $ be a decision tree over $U$ and
\begin{equation*}
\xi =v_{0},d_{0},v_{1},d_{1},\ldots ,v_{m},d_{m},v_{m+1}
\end{equation*}%
be a directed path from the root $v_{0}$ to a terminal node $v_{m+1}$ of $%
\Gamma $ (we call such path \emph{complete}). Define a subset $A(\xi )$ of
the set $A$ as follows. If $m=0$, then $A(\xi )=A$. Let $m>0$ and, for $%
i=1,\ldots ,m$, the node $v_{i}$ be labeled with the attribute $f_{j_{i}}$
and the edge $d_{i}$ be labeled with the number $\delta _{i}$. Then
\begin{equation*}
A(\xi )=\{a:a\in A,f_{j_{1}}(a)=\delta _{1},\ldots ,f_{j_{m}}(a)=\delta
_{m}\}.
\end{equation*}

The \emph{depth} of the decision tree $\Gamma $ is the maximum number of
working nodes in a complete path of $\Gamma $. Denote by $h(\Gamma )$ the
depth of $\Gamma $ and by $L(\Gamma )$ -- the number of nodes in $\Gamma $.

A decision tree over the information system $U$ is called a \emph{decision
tree over the problem }$z=(\nu ,f_{1},\ldots ,f_{n})$ if each working node
of $\Gamma $ is labeled with an attribute from the set $\{f_{1},\ldots
,f_{n}\}$.

The decision tree $\Gamma $ over $z$ solves the problem $z$ \emph{%
nondeterministically} if, for any object $a\in A$, there exists a complete
path $\xi $ of $\Gamma $ such that $a\in A(\xi )$ and, for each $a\in A$ and
each complete path $\xi $ such that $a\in A(\xi )$, the terminal node of $%
\xi $ is labeled with the number $z(a)$ (in this case, we can say that $%
\Gamma $ is a \emph{nondeterministic decision tree solving the problem} $z$%
). In particular, if the decision tree $\Gamma $ solves the problem $z$
nondeterministically, then, for each complete path $\xi $ of $\Gamma $,
either the set $A(\xi )$ is empty or the function $z(x)$ is constant on the
set $A(\xi )$. The decision tree $\Gamma $ over $z$ solves the problem $z$
\emph{deterministically} if $\Gamma $ is a deterministic decision tree,
which solves the problem $z$ nondeterministically (in this case, we can say
that $\Gamma $ is a \emph{deterministic decision tree solving the problem} $z
$).

Let $P(U)$ be the set of all problems over $U$. For a problem $z$ from $P(U)$%
, let $h_{U}^{ld}(z)$ be the minimum depth of a decision tree over $z$
solving the problem $z$ deterministically, $h_{U}^{la}(z)$ be the minimum
depth of a decision tree over $z$ solving the problem $z$
nondeterministically, $L_{U}^{ld}(z)$ be the minimum number of nodes in a
decision tree over $z$ solving the problem $z$ deterministically, and $%
L_{U}^{la}(z)$ be the minimum number of nodes in a decision tree over $z$
solving the problem $z$ nondeterministically.

We consider four functions defined on the set $\mathbb{N}$ in the following
way: $h_{U}^{ld}(n)=\max $ $h_{U}^{ld}(z)$, $h_{U}^{la}(n)=\max $ $%
h_{U}^{la}(z)$, $L_{U}^{ld}(n)=\max $ $L_{U}^{ld}(z)$, and $%
L_{U}^{la}(n)=\max $ $L_{U}^{la}(z)$, where the maximum is taken among all
problems $z$ over $U$ with $\dim z\leq n$. These functions describe how the
minimum depth and the minimum number of nodes of deterministic and
nondeterministic decision trees solving problems are growing in the worst
case with the growth of problem dimension. To describe possible types of
behavior of these four functions, we need to define some properties of
infinite binary information systems.

\begin{definition}
We will say that the information system $U=(A,F)$ satisfies the \emph{%
condition of reduction} if there exists $m\in \mathbb{N}$ such that, for
each compatible on $A$ system of equations
$
\{f_{1}(x)=\delta _{1},\ldots ,f_{r}(x)=\delta _{r}\},
$
where $r\in \mathbb{N}$, $f_{1},\ldots ,f_{r}\in F$ and $\delta
_{1},\ldots ,\delta _{r}\in \{0,1\}$, there exists a subsystem of this
system, which has the same set of solutions from $A$ and contains at most $m$
equations. In this case, we will say that $U$ satisfies the \emph{condition
of reduction with parameter} $m$.
\end{definition}

We now consider two examples of infinite binary information systems that satisfy the condition of reduction. These examples are close to ones considered in Section 3.4 of the book \cite{AbouEisha19}.

\begin{exmp}
 Let $d,t\in \mathbb{N}$,  $f_{1},\ldots ,f_{t}$
be functions from $\mathbb{R}^{d}$ to $\mathbb{R}$, where $\mathbb{R}$ is
the set of real numbers, and $s$ be a function from $\mathbb{R}$ to $\{0,1\}$
such that $s(x)=0$ if $x<0$ and $s(x)=1$ if $x\geq 0$. Then the infinite
binary information system $(\mathbb{R}^{d},F)$, where $F=\{s(f_{i}+c):i=1,%
\ldots ,t,c\in \mathbb{R}\}$, satisfies the condition of reduction with
parameter $2t$. If $f_{1},\ldots ,f_{t}$ are linear functions, then we deal
with attributes corresponding to $t$ families of parallel hyperplanes in $%
\mathbb{R}^{d}$ what is common for decision trees for datasets with $t$
numerical attributes only \cite{Breiman84}.
\end{exmp}

\begin{exmp}
Let $P$ be
 the Euclidean plane and $l$ be a straight line
(line in short) in the plane. This line divides the plane into two open
half-planes $H_{1}$ and $H_{2}$, and the line $l$. Two attributes correspond
to the line $l$. The first attribute takes value $0$ on points from $H_{1}$,
and value $1$ on points from $H_{2}$ and $l$. The second one takes value $0$
on points from $H_{2}$, and value $1$ on points from $H_{1}$ and $l$. We
denote by $\mathcal{L}$ the set of all attributes corresponding to lines in
the plane. Infinite binary information systems of the form $(P,L)$, where $%
L\subseteq \mathcal{L}$, are called \emph{linear} information systems.

Let $l$ be a line in the plane. Let us denote by $\mathcal{L}(l)$ the set of
all attributes corresponding to lines, which are parallel to $l$. Let $p$ be
a point in the plane. We denote by $\mathcal{L}(p)$ the set of all
attributes corresponding to lines, which pass through $p$. A set $C$ of
attributes from $\mathcal{L}$ is called a \emph{clone} if $C\subseteq
\mathcal{L}(l)$ for some line $l$ or $C\subseteq \mathcal{L}(p)$ for some
point $p$. In \cite{Moshkov07}, it was proved that a linear information
system $(P,L)$ satisfies the condition of reduction if and only if $L$ is
the union of a finite number of clones.
\end{exmp}

\begin{definition}
Let $U=(A,F)$ be an infinite binary information system.
A subset $\{f_{1},\ldots ,f_{m}\}$ of the set $F$ will be called \emph{%
independent} if, for any $\delta _{1},\ldots ,\delta _{m}\in \{0,1\}$, the
system of equations $\{f_{1}(x)=\delta _{1},\ldots ,f_{m}(x)=\delta _{m}\},$
has a solution from the set $A$. The empty set of attributes is independent
by definition.
\end{definition}

\begin{definition}
We define the parameter $I(U)$, which is called the \emph{independence
dimension} or I-\emph{dimension} of the information system $U$ (this notion
is similar to the notion of independence number of family of sets \cite%
{Naiman96}) as follows. If, for each $m\in \mathbb{N}$, the set $F$ contains
an independent subset of cardinality $m$, then $I(U)=\infty $. Otherwise, $%
I(U)$ is the maximum cardinality of an independent subset of the set $F$.
\end{definition}

We now consider examples of infinite binary information systems with finite I-dimension and with infinite I-dimension. More examples can be found in Lemmas \ref{L6.1}-\ref{L6.3}.

\begin{exmp}
Let $m,t\in \mathbb{N}$. We denote by $Pol(m)$ the set of all polynomials,
which have integer coefficients and depend on variables $x_{1},\ldots ,x_{m}$%
. We denote by $Pol(m,t)$ the set of all polynomials from $Pol(m)$ such that
the degree of each polynomial is at most $t$. We define infinite binary
information systems $U(m)$ and $U(m,t)$ as follows: $U(m)=(\mathbb{R}%
^{m},F(m))$ and $U(m,t)=(\mathbb{R}^{m},F(m,t))$, where  $%
F(m)=\{s(p):p\in Pol(m)\}$, $F(m,t)=\{s(p):p\in Pol(m,t)\}$, and $s(x)=0$
if $x<0$ and $s(x)=1$ if $x\geq 0$. One can show that the system $U(m)$ has
infinite I-dimension and the system $U(m,t)$ has finite I-dimension.
\end{exmp}

We now consider four statements that describe possible types of behavior of
functions $h_{U}^{ld}(n)$, $h_{U}^{la}(n)$, $L_{U}^{ld}(n)$, and $%
L_{U}^{la}(n)$. The next statement follows immediately from Theorem 4.3 \cite%
{Moshkov05}.

\begin{proposition}
\label{P1}For any infinite binary information system $U$, the function $%
h_{U}^{ld}(n)$ has one of the following two types of behavior:

\textrm{(LOG)} If the system $U$ satisfies the condition of reduction, then $%
h_{U}^{ld}(n)=\Theta (\log n)$.

\textrm{(LIN)} If the system $U$ does not satisfy the condition of
reduction, then $h_{U}^{ld}(n)=n$ for any $n\in \mathbb{N}$.
\end{proposition}

The next statement follows immediately from Theorem 8.2 \cite{Moshkov11}.

\begin{proposition}
\label{P2}For any infinite binary information system $U=(A,F)\,$, the
function $h_{U}^{la}(n)$ has one of the following two types of behavior:

\textrm{(CON)} If the system $U$ satisfies the condition of reduction, then $%
h_{U}^{la}(n)=O(1)$.

\textrm{(LIN)} If the system $U$ does not satisfy the condition of reduction,
then $h_{U}^{la}(n)=n$ for any $n\in \mathbb{N}$.
\end{proposition}

\begin{proposition}
\label{P3}For any infinite binary information system $U$, the function $%
L_{U}^{ld}(n)$ has one of the following two types of behavior:

\textrm{(POL)} If the system $U$ has finite I-dimension, then for any $n\in
\mathbb{N}$,
\begin{equation*}
2(n+1)\leq L_{U}^{ld}(n)\leq 2(4n)^{I(U)}.
\end{equation*}

\textrm{(EXP)} If the system $U$ has infinite I-dimension, then for any $%
n\in \mathbb{N}$,
\begin{equation*}
L_{U}^{ld}(n)=2^{n+1}.
\end{equation*}
\end{proposition}

\begin{proposition}
\label{P4}For any infinite binary information system $U$ and any $n\in
\mathbb{N}$,%
\begin{equation*}
L_{U}^{la}(n)=L_{U}^{ld}(n).
\end{equation*}
\end{proposition}

Let $U$ be an infinite binary information system. Proposition \ref{P1}
allows us to correspond to the function $h_{U}^{ld}(n)$ its type of
behavior from the set $\{\mathrm{LOG},\mathrm{LIN}\}$. Proposition \ref{P2}
allows us to correspond to the function $h_{U}^{la}(n)$ its type of behavior
from the set $\{\mathrm{CON},\mathrm{LIN}\}$. Propositions \ref{P3} and \ref%
{P4} allow us to correspond to each of the functions $L_{U}^{ld}(n)$ and $%
L_{U}^{la}(n)$ its type of behavior from the set $\{\mathrm{POL},\mathrm{EXP}%
\}$. A tuple obtained from the tuple
\begin{equation*}
(h_{U}^{ld}(n),h_{U}^{la}(n),L_{U}^{ld}(n),L_{U}^{la}(n))
\end{equation*}%
by replacing functions with their types of behavior is called the \emph{local type} of
the information system $U$. We now describe all possible local types of infinite
binary information systems.

\begin{theorem}
\label{T1}For any infinite binary information system, its local type coincides
with one of the rows of Table \ref{tab1}. Each row of Table \ref{tab1} is
the local type of some infinite binary information system.
\end{theorem}

\begin{table}[h]
\centering
\caption{Possible local types of infinite binary information systems}
\begin{tabular}{|l|llll|} \hline
& $h_{U}^{ld}(n)$ & $h_{U}^{la}(n)$ & $L_{U}^{ld}(n)$ & $%
L_{U}^{la}(n)$ \\\hline
1 & $\mathrm{LOG}$ & $\mathrm{CON}$ & $\mathrm{POL}$ & $\mathrm{POL}$
\\
2 & $\mathrm{LIN}$ & $\mathrm{LIN}$ & $\mathrm{POL}$ & $\mathrm{POL}$ \\
3 & $\mathrm{LIN}$ & $\mathrm{LIN}$ & $\mathrm{EXP}$ & $\mathrm{EXP}$ \\ \hline
\end{tabular}
  \label{tab1}
\end{table}

For $i=1,2,3$, we denote by $W_{i}^{l}$ the class of all infinite binary
information systems, whose local type coincides with the $i$th row of Table \ref%
{tab1}. We now study for each of these complexity classes joint behavior of
the depth and number of nodes in decision trees solving problems.

For a given infinite binary information system $U$, we will consider pairs
of functions $(\varphi ,\psi )$ such that, for any problem $z$ over $U$,
there exists a deterministic decision tree over $z$ solving $z$ with the
depth at most $\varphi (\dim z)$ and the number of nodes at most $\psi (\dim
z)$. We will study such pairs and will call them boundary $ld$-pairs.

\begin{definition}
A pair of functions $(\varphi ,\psi )$, where $\varphi :\mathbb{N}%
\rightarrow \mathbb{N}\cup \{0\}$ and $\psi :\mathbb{N}\rightarrow \mathbb{N}%
\cup \{0\}$, will be called a \emph{boundary $ld$-pair} of the information
system $U$ if, for any problem $z$ over $U$, there exists a decision tree $%
\Gamma $ over $z$, which solves the problem $z$ deterministically and for
which $h(\Gamma )\leq \varphi (n)$ and $L(\Gamma )\leq \psi (n)$, where $%
n=\dim z$.
\end{definition}

We are interested in finding boundary $ld$-pairs with functions that grow as
slowly as possible.
It is clear that, for any boundary $ld$-pair $(\varphi ,\psi )$ of the
information system $U$, the following inequalities hold: $\varphi (n)\geq
h_{U}^{ld}(n)$ and $\psi (n)\geq L_{U}^{ld}(n)$. So the best possible
situation is when $(h_{U}^{ld},L_{U}^{ld})$ is a boundary $ld$-pair of $U$.

\begin{definition}
An information system $U$ will be called $ld$-\emph{reachable} if the pair $%
(h_{U}^{ld},L_{U}^{ld})$ is a boundary $ld$-pair of the system $U$.
\end{definition}

We now consider similar notions for nondeterministic decision trees: the
notion of boundary $la$-pair and the notion of $la$-reachable information
system.

\begin{definition}
A pair of functions $(\varphi ,\psi )$, where $\varphi :\mathbb{N}%
\rightarrow \mathbb{N}\cup \{0\}$ and $\psi :\mathbb{N}\rightarrow \mathbb{N}%
\cup \{0\}$, will be called a \emph{boundary $la$-pair} of the information
system $U$ if, for any problem $z$ over $U$, there exists a decision tree $%
\Gamma $ over $z$, which solves the problem $z$ nondeterministically and for
which $h(\Gamma )\leq \varphi (n)$ and $L(\Gamma )\leq \psi (n)$, where $%
n=\dim z$.
\end{definition}

\begin{definition}
An information system $U$ will be called $la$-\emph{reachable} if the pair $%
(h_{U}^{la},L_{U}^{la})$ is a boundary $la$-pair of the system $U$.
\end{definition}

Note that for deterministic decision trees, the best situation is when the
considered information system is $ld$-reachable and for nondeterministic
decision trees -- when the information system is $la$-reachable.

Each information system from the classes $W_{1}^{l},W_{2}^{l}$, and $W_{3}^{l}$ is $%
ld$-reachable.
Each information system from the classes $W_{2}^{l}$ and $W_{3}^{l}$ is $la$%
-reachable. Each information system from the class $W_{1}^{l}$ is not $la$%
-reachable. For all information systems $U$, which are not $la$-reachable,
we find nontrivial boundary $la$-pairs that are sufficiently close to $%
(h_{U}^{la},L_{U}^{la})$.

The obtained results are related to time-space trade-off for deterministic
and nondeterministic decision trees. Details can be found in the following
three theorems.

\begin{theorem}
\label{T2} Let $U$ be an information system from the class $W_{1}^{l}$. Then

\textrm{(a)} The system $U$ is $ld$-reachable.

\textrm{(b)} The system $U$ is not $la$-reachable and there exists $m\in
\mathbb{N}$ such that $$(m,(m+1)L_{U}^{la}(n)/2+1)$$ is a boundary $la$-pair
of the system $U$.
\end{theorem}

\begin{theorem}
\label{T3}Let $U$ be an information system from the class $W_{2}^l$. Then

\textrm{(a)} The system $U$ is $ld$-reachable.

\textrm{(b)} The system $U$ is $la$-reachable.
\end{theorem}

\begin{theorem}
\label{T4}Let $U$ be an information system from the class $W_{3}^l$. Then

\textrm{(a)} The system $U$ is $ld$-reachable.

\textrm{(b)} The system $U$ is $la$-reachable.
\end{theorem}

Table \ref{tab3} summarizes Theorems \ref{T1}-\ref{T4}. The first column
contains the name of the complexity class. The next four columns describe
the local type of information systems from this class. The last two columns
\textquotedblleft $ld$-pairs\textquotedblright\ and \textquotedblleft $la$%
-pairs\textquotedblright\ contain information about boundary $ld$-pairs and
boundary $la$-pairs for information systems from the considered class:
\textquotedblleft $ld$-reachable\textquotedblright\ means that all
information systems from the class are $ld$-reachable, \textquotedblleft $la$%
-reachable\textquotedblright\ means that all information systems from the
class are $la$-reachable, Th. \ref{T2} (b) is a link to the corresponding
statement Theorem \ref{T2} (b).

\begin{table}[h]
\centering
\caption{Summary of Theorems \ref{T1}-\ref{T4}}
\begin{tabular}{|l|llll|cc|} \hline
 & $h_{U}^{ld}(n)$ & $h_{U}^{la}(n)$ & $L_{U}^{ld}(n)$ & $%
L_{U}^{la}(n)$ & $ld$-pairs & $la$-pairs \\\hline
 $W_{1}^l$ & $\mathrm{LOG}$ & $\mathrm{CON}$ & $\mathrm{POL}$ & $%
\mathrm{POL}$ & $ld$-reachable & Th. \ref{T2} (b) \\
$W_{2}^l$ & $\mathrm{LIN}$ & $\mathrm{LIN}$ & $\mathrm{POL}$ & $\mathrm{POL}$
& $ld$-reachable & $la$-reachable \\
$W_{3}^l$ & $\mathrm{LIN}$ & $\mathrm{LIN}$ & $\mathrm{EXP}$ & $\mathrm{EXP}$
& $ld$-reachable & $la$-reachable \\ \hline
\end{tabular}
  \label{tab3}
\end{table}

\section{Proofs of Propositions \protect\ref{P3} and \protect\ref{P4}}

\label{S3}

In this section, we consider a number of auxiliary statements and prove the two
mentioned propositions.

 Let $\Gamma $ be a decision tree
over an infinite binary information system $U=(A,F)$ and $d$ be an edge of $\Gamma $
entering a node $w$. We denote by $\Gamma (d)$ a subtree of $\Gamma $, whose
root is the node $w$. We say that a complete path $\xi $ of $\Gamma $ is
\emph{realizable} if $A(\xi )\neq \emptyset $.

\begin{lemma}
\label{L1}Let $U=(A,F)$ be an infinite binary information system, $z=(\nu
,f_{1},\ldots ,f_{n})$ be a problem over $U$, and $\Gamma $ be a decision
tree over $z$, which solves the problem $z$ deterministically and for which $%
L(\Gamma )=L_{U}^{ld}(z)$. Then

\textrm{(a)} For each node of $\Gamma $, there exists a realizable complete
path that passes through this node.

\textrm{(b)} Each working node of $\Gamma $ has two edges leaving this node.
\end{lemma}

\begin{proof}
(a) It is clear that there exists at least one realizable complete path that
passes through the root of $\Gamma $. Let us assume that $w$ is a node of $%
\Gamma $ different from the root and such that there is no a realizable
complete path, which passes through $w$. Let $d$ be an edge entering the
node $w$. We remove from $\Gamma $ the edge $d$ and the subtree $\Gamma (d)$%
. As a result, we obtain a decision tree $\Gamma ^{\prime }$ over $z$, which solves $z
$ deterministically and for which $L(\Gamma ^{\prime })<L(\Gamma )$ but this
is impossible by definition of $\Gamma $.

(b) Let us assume that in $\Gamma $ there exists a working node $w$, which
has only one leaving edge $d$ entering a node $w_{1}$. We remove from $%
\Gamma $ the node $w$ and the edge $d$ and connect the edge entering the
node $w$ to the node $w_{1}$. As a result, we obtain a decision tree $\Gamma
^{\prime }$ over $z$, which solves the problem $z$ deterministically and for which $%
L(\Gamma ^{\prime })<L(\Gamma )$ but this is impossible by definition of $%
\Gamma $. \qed
\end{proof}

Let $U$ be an infinite binary information system, $\Gamma $ be a decision
tree over $U$, and $d$ be an edge of $\Gamma $. The subtree $\Gamma (d)$
will be called \emph{full} if there exist edges $d_{1},\ldots ,d_{m}$ in $%
\Gamma (d)$ such that the removal of these edges and subtrees $\Gamma
(d_{1}),\ldots ,\Gamma (d_{m})$ transforms the subtree $\Gamma (d)$ into a
tree $G$ such that each terminal node of $G$ is a terminal node of $\Gamma $%
, and exactly two edges labeled with the numbers $0$ and $1$ respectively
leave each working node of $G$.

\begin{lemma}
\label{L2}Let $U=(A,F)$ be an infinite binary information system, $z=(\nu
,f_{1},\ldots ,f_{n})$ be a problem over $U$, and $\Gamma $ be a decision
tree over $z$, which solves the problem $z$ nondeterministically and for
which $L(\Gamma )=L_{U}^{la}(z)$. Then

\textrm{(a)} For each node of $\Gamma $, there exists a realizable complete
path that passes through this node.

\textrm{(b)} If a working node $w$ of $\Gamma $ has $m$ leaving edges $%
d_{1},\ldots ,d_{m}$ labeled with the same number and $m\geq 2$, then the
subtrees $\Gamma (d_{1}),\ldots ,\Gamma (d_{m})$ are not full.

\textrm{(c)} If the root $r$ of $\Gamma $ has $m$ leaving edges $%
d_{1},\ldots ,d_{m}$ and $m\geq 2$, then the subtrees $\Gamma (d_{1}),\ldots
,\Gamma (d_{m})$ are not full.
\end{lemma}

\begin{proof}
(a) The proof of item (a) is almost identical to the proof of item (a) of
Lemma \ref{L1}.

(b) Let $w$ be a working node of $\Gamma $, which has $m$ leaving edges $%
d_{1},\ldots ,d_{m}$ labeled with the same number, $m\geq 2$, and at least
one of the subtrees $\Gamma (d_{1}),\ldots ,$ $\Gamma (d_{m})$ is full. For the
definiteness, we assume that $\Gamma (d_{1})$ is full. Remove from $\Gamma $
the edges $d_{2},\ldots ,d_{m}$ and subtrees $\Gamma (d_{2}),\ldots ,\Gamma
(d_{m})$. We now show that the obtained tree $\Gamma ^{\prime }$ solves the
problem $z$ nondeterministically. Assume the contrary. Then there exists an
object $a\in A$ such that, for each complete path $\xi $ with $a\in A(\xi )$%
, the path $\xi $ passes through one of the edges $d_{2},\ldots ,d_{m}$ but
it is not true. Let $\xi $ be a complete path such that $a\in A(\xi )$.
Then, according to the assumption, this path passes through the node $w$.
Let $\xi ^{\prime }$ be the part of this path from the root of $\Gamma $ to
the node $w$. Since the edges $d_{1},\ldots ,d_{m}$ are labeled with the
same number and $\Gamma (d_{1})$ is a full subtree, we can find in $\Gamma
(d_{1})$ the continuation of $\xi ^{\prime }$ to a terminal node of $\Gamma
(d_{1})$ such that the obtained complete path $\xi ^{\prime \prime }$ of $%
\Gamma $ satisfies the condition $a\in A(\xi ^{\prime \prime })$. Hence $%
\Gamma ^{\prime }$ is a decision tree over $z$, which solves the problem $z$
nondeterministically and for which $L(\Gamma ^{\prime })<L(\Gamma )$, but
this is impossible by definition of $\Gamma $.

(c) Item (c) can be proven in the same way as item (b). \qed
\end{proof}

We now consider two statements about classes of decision trees  proved in \cite{Moshkov23}. Let $%
\Gamma $ be a decision tree. We denote by $L_{t}(\Gamma )$ the number of
terminal nodes in $\Gamma $ and by $L_{w}(\Gamma )$ the number of working
nodes in $\Gamma $. It is clear that $L(\Gamma )=1+L_{t}(\Gamma
)+L_{w}(\Gamma )$.

Let $U$ be an infinite binary information systems. We denote by $G_{d}(U)$
the set of all deterministic decision trees over $U$ and by $G_{d}^{2}(U)$
the set of all decision trees from $G_{d}(U)$ such that each working node of
the tree has two leaving edges.

\begin{lemma}[Lemma 14 from \cite{Moshkov23}]
\label{L3}Let $U$ be an infinite binary information system. Then

\textrm{(a)} If $\Gamma \in G_{d}^{2}(U)$, then $L_{w}(\Gamma )=L_{t}(\Gamma
)-1$.

\textrm{(b)} If $\Gamma \in G_{d}(U)\setminus G_{d}^{2}(U)$, then $%
L_{w}(\Gamma )>L_{t}(\Gamma )-1$.
\end{lemma}

We denote by $G_{a}^{f}(U)$ the set of all decision trees $\Gamma $ over $U$
that satisfy the following conditions: (i) if a working node of $\Gamma $
has $m$ leaving edges $d_{1},\ldots ,d_{m}$ labeled with the same number and
$m\geq 2$, then the subtrees $\Gamma (d_{1}),\ldots ,\Gamma (d_{m})$ are not
full, and (ii) if the root of $\Gamma $ has $m$ leaving edges $d_{1},\ldots
,d_{m}$ and $m\geq 2$, then the subtrees $\Gamma (d_{1}),\ldots ,\Gamma
(d_{m})$ are not full. One can show that $G_{d}^{2}(U)\subseteq
G_{d}(U)\subseteq G_{a}^{f}(U)$.

\begin{lemma}[Lemma 15 from \cite{Moshkov23}]
\label{L4}Let $U$ be an infinite binary information system. If $\Gamma \in $
$G_{a}^{f}(U)\setminus G_{d}^{2}(U)$, then $L_{w}(\Gamma )>L_{t}(\Gamma )-1$.
\end{lemma}

Let $U=(A,F)$ be an infinite binary information system. For $f_{1},\ldots
,f_{n}\in F$, we denote by $N_{U}(f_{1},\ldots ,f_{n})$ the number of $n$%
-tuples $(\delta _{1},\ldots ,\delta _{n})\in \{0,1\}^{n}$ for which the
system of equations
\begin{equation*}
\{f_{1}(x)=\delta _{1},\ldots ,f_{n}(x)=\delta _{n}\}
\end{equation*}%
has a solution from $A$. For $n\in \mathbb{N}$, denote
\begin{equation*}
N_{U}(n)=\max \{N_{U}(f_{1},\ldots ,f_{n}):f_{1},\ldots ,f_{n}\in F\}.
\end{equation*}

It is clear that, for any $m,n\in \mathbb{N}$, if $m\leq n$ then $%
N_{U}(m)\leq N_{U}(n)$.

\begin{proposition}
\label{P5} Let $U=(A,F)$ be an infinite binary information system. Then, for
any $n\in \mathbb{N}$,
\begin{equation*}
L_{U}^{la}(n)=L_{U}^{ld}(n)=2N_{U}(n).
\end{equation*}
\end{proposition}

\begin{proof}
Let $z=(\nu ,f_{1},\ldots ,f_{m})$ be a problem over $U$ and $m\leq n$. Let $%
\Gamma $ be a decision tree over $z$, which solves the problem $z$
deterministically and for which $L(\Gamma )=L_{U}^{ld}(z)$. From Lemma \ref{L1} it follows that each working node of $\Gamma $ has two edges leaving this node and, for each
node of $\Gamma $, there exists a realizable complete path that passes
through this node. Let $\xi _{1}$ and $\xi _{2}$ be different complete paths
in $\Gamma $, $a_{1}\in A(\xi _{1})$, and $a_{2}\in A(\xi _{2})$. It is easy
to show that $(f_{1}(a_{1}),\ldots ,f_{m}(a_{1}))\neq (f_{1}(a_{2}),\ldots
,f_{m}(a_{2}))$. Therefore $L_{t}(\Gamma )\leq N_{U}(f_{1},\ldots
,f_{m})\leq N_{U}(n)$. It is clear that $\Gamma \in G_{d}^{2}(U)$. By Lemma %
\ref{L3}, $L_{w}(\Gamma )=L_{t}(\Gamma )-1$. Hence $L(\Gamma )\leq 2N_{U}(n)$%
. Taking into account that $z$ is an arbitrary problem over $U$ with $\dim
z\leq n$ we obtain
\begin{equation*}
L_{U}^{ld}(n)\leq 2N_{U}(n).
\end{equation*}

Since any decision tree solving the problem $z$ deterministically solves it
nondeterministically we obtain%
\begin{equation*}
L_{U}^{la}(n)\leq L_{U}^{ld}(n).
\end{equation*}%
We now show that $2N_{U}(n)\leq L_{U}^{la}(n)$. Let us consider a problem $%
z=(\nu ,f_{1},\ldots ,f_{n})$ over $U$ such that
\begin{equation*}
N_{U}(f_{1},\ldots ,f_{n})=N_{U}(n)
\end{equation*}%
and, for any $\bar{\delta}_{1},\bar{\delta}_{2}\in \{0,1\}^{n}$, if $\bar{%
\delta}_{1}\neq \bar{\delta}_{2}$, then $\nu (\bar{\delta}_{1})\neq \nu (%
\bar{\delta}_{2})$. Let $\Gamma $ be a decision tree over $z$, which solves
the problem $z$ nondeterministically and for which $L(\Gamma )=L_{U}^{la}(z)$%
. By Lemma \ref{L2}, $\Gamma \in G_{a}^{f}(U)$. Using Lemmas \ref{L3} and %
\ref{L4} we obtain $L_{w}(\Gamma )\geq L_{t}(\Gamma )-1$. It is clear that $%
L_{t}(\Gamma )\geq N_{U}(f_{1},\ldots ,f_{n})=N_{U}(n)$. Therefore $L(\Gamma
)\geq 2N_{U}(n)$, $L_{U}^{la}(z)\geq 2N_{U}(n)$, and $L_{U}^{la}(n)\geq
2N_{U}(n)$. \qed
\end{proof}

The next statement follows directly from Lemmas 5.1 and 5.2 \cite{Moshkov05}
and the evident inequality $N_{U}(n)\leq 2^{n}$, which is true for any
infinite binary information system $U$. The proof of Lemma 5.1 from \cite%
{Moshkov05} is based on Theorems 4.6 and 4.7 from the same monograph that
are similar to results obtained in \cite{Sauer72,Shelah72}.

\begin{proposition}
\label{P6} For any infinite binary information system $U$, the function $%
N_{U}(n)$ has one of the following two types of behavior:

\textrm{(POL)} If the system $U$ has finite I-dimension, then for any $n\in
\mathbb{N}$,
\begin{equation*}
n+1\leq N_{U}(n)\leq (4n)^{I(U)}.
\end{equation*}

\textrm{(EXP)} If the system $U$ has infinite I-dimension, then for any $%
n\in \mathbb{N}$,
\begin{equation*}
N_{U}(n)=2^{n}.
\end{equation*}
\end{proposition}

We now prove Propositions \ref{P3} and \ref{P4}.

\begin{proof}[Proof of Proposition \protect\ref{P3}]
The statement of the proposition follows immediately from Propositions \ref%
{P5} and \ref{P6}. \qed
\end{proof}

\begin{proof}[Proof of Proposition \protect\ref{P4}]
The statement of the proposition follows immediately from Proposition \ref%
{P5}. \qed
\end{proof}

\section{Proof of Theorem \protect\ref{T1}}

\label{S4}

First, we prove five auxiliary statements.

\begin{lemma}
\label{L6.0}Let $U=(A,F)$ be an infinite binary information system, which
has infinite I-dimension. Then $U$ does not satisfy the condition of
reduction.
\end{lemma}

\begin{proof}
Let us assume the contrary: $U$ satisfies the condition of reduction. Then $%
U $ satisfies the condition of reduction with parameter $m$ for some $m\in
\mathbb{N}$. Since $I(U)=\infty $, there exists an independent subset $%
\{f_{1},\ldots ,f_{m+1}\}$ of the set $F$. It is clear that the system of
equations%
\begin{equation*}
S=\{f_{1}(x)=0,\ldots ,f_{m+1}(x)=0\}
\end{equation*}%
is compatible on $A$ and each proper subsystem of the system $S$ has the set
of solutions different from the set of solutions of $S$. Therefore $U$ does
not satisfy the condition of reduction with parameter $m$. \qed
\end{proof}

\begin{lemma}
\label{L5}For any infinite binary information system, its local type coincides
with one of the rows of Table \ref{tab1}.
\end{lemma}

\begin{proof}
To prove this statement we fill Table \ref{tab2}. In the first column
\textquotedblleft I-dim.\textquotedblright\ we have either \textquotedblleft
Fin\textquotedblright\ or \textquotedblleft Inf\textquotedblright :
\textquotedblleft Fin\textquotedblright\ if the considered information
system has finite I-dimension and \textquotedblleft Inf\textquotedblright\
if the considered information system has infinite I-dimension. In the second
column \textquotedblleft Reduct.\textquotedblright , we have either
\textquotedblleft Yes\textquotedblright\ or \textquotedblleft
No\textquotedblright : \textquotedblleft Yes\textquotedblright\ if the
considered information system satisfies the condition of reduction and
\textquotedblleft No\textquotedblright\ otherwise.

By Lemma \ref{L6.0}, if an information system has infinite I-dimension, then
this information system does not satisfy the condition of reduction. It
means that there are only three possible tuples of values of the considered
two parameters of information systems, which correspond to the three rows of
Table \ref{tab2}. The values of the considered two parameters define the
types of behavior of functions $h_{U}^{ld}(n)$, $h_{U}^{la}(n)$, $%
L_{U}^{ld}(n)$, and $L_{U}^{la}(n)$ according to Propositions \ref{P1}-\ref%
{P4}. We see that the set of possible tuples of values in the last four
columns coincides with the set of rows of Table \ref{tab1}. \qed
\end{proof}

\begin{table}[h]
\centering
\caption{Parameters and local types of infinite binary information systems}
\begin{tabular}{|ll|llll|}
\hline I-dim. & Reduct. & $h_{U}^{ld}(n)$ & $h_{U}^{la}(n)$ & $%
L_{U}^{ld}(n)$ & $L_{U}^{la}(n)$ \\
\hline Fin & Yes & $\mathrm{LOG}$ & $\mathrm{CON}$ & $\mathrm{POL}$ & $%
\mathrm{POL}$ \\
Fin & No & $\mathrm{LIN}$ & $\mathrm{LIN}$ & $\mathrm{POL}$ & $\mathrm{POL}$
\\
Inf & No & $\mathrm{LIN}$ & $\mathrm{LIN}$ & $\mathrm{EXP}$ & $\mathrm{EXP}$
\\
\hline 
\end{tabular}
  \label{tab2}
\end{table}

For each row of Table \ref{tab1}, we consider an example of infinite binary
information system, whose local type coincides with this row.

For any $i\in \mathbb{N}$, we define two functions $p_{i}:\mathbb{N}%
\rightarrow \{0,1\}$ and $l_{i}:\mathbb{N}\rightarrow \{0,1\}$. Let $j\in
\mathbb{N}$. Then $p_{i}(j)=1$ if and only if $j=i$, and $l_{i}(j)=1$ if and
only if $j>i$.

Define an information system $U_{1}=(A_{1},F_{1})$ as follows: $A_{1}=%
\mathbb{N}$ and $F_{1}=\{l_{i}:i\in \mathbb{N}\}$.

\begin{lemma}
\label{L6.1} The information system $U_{1}$ belongs to the class $W_{1}^{l}$%
, $h_{U_{1}}^{ld}(n)=\lceil \log _{2}(n$ $+1)\rceil $, $h_{U_{1}}^{la}(1)=1$
and $h_{U_{1}}^{la}(n)=2$ if $n>1$, $L_{U_{1}}^{ld}(n)=2(n+1)$, and $%
L_{U_{1}}^{la}(n)=2(n+1)$ for any $n\in \mathbb{N}$. This information system
satisfies the condition of reduction with parameter $2$ and has finite
I-dimension equal $1$.
\end{lemma}

\begin{proof}
It is easy to show that $N_{U_{1}}(n)=n+1$ for any $n\in \mathbb{N}$. Using
Proposition \ref{P5} we obtain $L_{U_{1}}^{ld}(n)=L_{U_{1}}^{la}(n)=2(n+1)$
for any $n\in \mathbb{N}$. Let $n\in \mathbb{N}$. Consider a problem $z=(\nu
,l_{1},\ldots ,l_{n})$ over $U_{1}$ such that, for each $\bar{\delta}_{1},%
\bar{\delta}_{2}\in \{0,1\}^{n}$ with $\bar{\delta}_{1}\neq \bar{\delta}_{2}$%
, $\nu (\bar{\delta}_{1})\neq \nu (\bar{\delta}_{2})$. It is clear that $%
N_{U_{1}}(l_{1},\ldots ,l_{n})=n+1$. Therefore each decision tree $\Gamma $
over $z$ that solves the problem $z$ deterministically has at least $n+1$
terminal nodes. One can show that the number of terminal nodes in $\Gamma $
is at most $2^{h(\Gamma )}$. Hence $n+1\leq 2^{h(\Gamma )}$ and $\log
_{2}(n+1)\leq h(\Gamma )$. Since $h(\Gamma )$ is an integer, $\lceil \log
_{2}(n+1)\rceil \leq h(\Gamma )$. Thus, $h_{U_{1}}^{ld}(n)\geq \lceil \log
_{2}(n+1)\rceil $. Set $m=\lceil \log _{2}(n+1)\rceil $. Then $n\leq 2^{m}-1$%
. One can show that $h_{U_{1}}^{ld}(2^{m}-1)\leq m$ (the construction of an
appropriate decision tree is based on an analog of binary search, and we use
only attributes from the problem description) and $h_{U_{1}}^{ld}(n)\leq
h_{U_{1}}^{ld}(2^{m}-1)$. Therefore $h_{U_{1}}^{ld}(n)\leq \lceil \log
_{2}(n+1)\rceil $ and $h_{U_{1}}^{ld}(n)=\lceil \log _{2}(n+1)\rceil $. It
is clear that $h_{U_{1}}^{la}(1)=1$. Let $n\geq 2$, $z=(\nu ,f_{1},\ldots
,f_{n})$ be an arbitrary problem over $U_{1}$ and $l_{i_{1}},\ldots
,l_{i_{m}}$ be all pairwise different attributes from the set $%
\{f_{1},\ldots ,f_{n}\}$ ordered such that $i_{1}<\ldots <i_{m}$. Then these
attributes divide the set $\mathbb{N}$ into $m+1$ nonempty domains that are
sets of solutions on $\mathbb{N}$ of the following systems of equations: $%
\{l_{i_{1}}(x)=0\}$, $\{l_{i_{1}}(x)=1,l_{i_{2}}(x)=0\}$, \ldots , $%
\{l_{i_{m-1}}(x)=1,l_{i_{m}}(x)=0\}$, $\{l_{i_{m}}(x)=1\}$. The value $z(x)$
is constant in each of the considered domains. Using these facts it is easy
to show that there exists a decision tree $\Gamma $ over $z$, which solves
the problem $z$ nondeterministically and for which $h(\Gamma )=2$ if $m\geq 2
$. Therefore $h_{U_{1}}^{la}(n)\leq 2$. One can show that there exists a
problem $z$ over $U_{1}$ such that $\dim z=n$ and $h_{U_{1}}^{la}(z)\geq 2$.
Therefore $h_{U_{1}}^{la}(n)=2$.

Since the function $h_{U_{1}}^{ld}$ has the type of behavior \textrm{LOG},
the information system $U_{1}$ belongs to the class $W_{1}^{l}$ -- see Lemma \ref{L5} and Table %
\ref{tab1}. One can show that the information system $U_{1}$ satisfies the
condition of reduction with parameter $2$ and has finite I-dimension equal $1
$. \qed
\end{proof}

Define an information system $U_{2}=(A_{2},F_{2})$ as follows: $A_{2}=%
\mathbb{N}$ and $F_{2}=\{p_{i}:i\in \mathbb{N}\}$.

\begin{lemma}
\label{L6.2}The information system $U_{2}$ belongs to the class $W_{2}^{l}$,
$h_{U_{2}}^{ld}(n)=n$, $h_{U_{2}}^{la}(n)=n$, $L_{U_{2}}^{ld}(n)=2(n+1)$,
and $L_{U_{2}}^{la}(n)=2(n+1)$ for any $n\in \mathbb{N}$. This information
system does not satisfy the condition of reduction and has finite
I-dimension equal $1$.
\end{lemma}

\begin{proof}
It is easy to show that $N_{U_{2}}(n)=n+1$ for any $n\in \mathbb{N}$. Using
Proposition \ref{P5}, we obtain $L_{U_{2}}^{ld}(n)=L_{U_{2}}^{la}(n)=2(n+1)$
for any $n\in \mathbb{N}$.

Let $n\in \mathbb{N}$.
It is clear that the system of
equations%
\begin{equation*}
S_n=\{p_{1}(x)=0,\ldots ,p_{n}(x)=0\}
\end{equation*}%
is compatible on $A_2$ and each proper subsystem of the system $S_n$ has the set
of solutions different from the set of solutions of $S_n$. Therefore $U_2$ does not satisfy the condition of reduction. Using
Propositions \ref{P1} and \ref{P2}, we obtain $h_{U_{2}}^{ld}(n)=h_{U_{2}}^{la}(n)=n$
for any $n\in \mathbb{N}$.

Since the function $h_{U_{2}}^{ld}$  has the type of
behavior \textrm{LIN} and the function $L_{U_{2}}^{ld}$
has the type of behavior \textrm{POL}, the information system $U_{2}$
belongs to the class $W_{2}^{l}$ -- see Lemma \ref{L5} and  Table \ref{tab1}. One can show that
this information system has finite I-dimension equal $1$. \qed
\end{proof}

Define an information system $U_{3}=(A_{3},F_{3})$ as follows: $A_{3}=%
\mathbb{N}$ and $F_{3}$ is the set of all functions from $\mathbb{N}$ to $%
\{0,1\}$.

\begin{lemma}
\label{L6.3}The information system $U_{3}$ belongs to the class $W_{3}^{l}$, $%
h_{U_{3}}^{ld}(n)=n$, $h_{U_{3}}^{la}(n)=n$, $L_{U_{3}}^{ld}(n)=2^{n+1}$,
and $L_{U_{3}}^{la}(n)=2^{n+1}$ for any $n\in \mathbb{N}$. This information
system does not satisfy the condition of reduction and has infinite I-dimension.
\end{lemma}

\begin{proof}
It is easy to show that the information system $U_{3}$ has infinite
I-dimension. Using Lemma \ref{L6.0}, we obtain that the system $U_{3}$ does
not satisfy the condition of reduction. Let $n\in \mathbb{N}$. By
Propositions \ref{P1} and \ref{P2}, $h_{U_{3}}^{ld}(n)=
h_{U_{3}}^{la}(n)=n$. By Propositions \ref{P3} and \ref{P4}, $L_{U_{3}}^{ld}(n)=L_{U_{3}}^{la}(n)=2^{n+1}$.

Since the function  $L_{U_{3}}^{ld}$
 has the type of behavior \textrm{EXP}, the information system $U_{3}$
belongs to the class $W_{3}^{l}$ -- see Lemma \ref{L5} and  Table \ref{tab1}. \qed
\end{proof}

\begin{proof}[Proof of Theorem \protect\ref{T1}]
The statements of the theorem follow from Lemmas \ref{L5}-\ref{L6.3}. \qed
\end{proof}

\section{Proofs of Theorems \protect\ref{T2}-\protect\ref{T4}}

\label{S5}

First, we prove a number of auxiliary statements.

\begin{lemma}
\label{L7}Let $U=(A,F)$ be an infinite binary information system. Then the
information system $U$ is $ld$-reachable.
\end{lemma}

\begin{proof}
Let $z=(\nu ,f_{1},\ldots ,f_{n})$ be a problem over $U$. Then there
exists a decision tree $\Gamma $ over $z$, which solves this problem
deterministically and whose depth is at most $h_{U}^{ld}(n)$. By removal of
some nodes and edges from $\Gamma $, we can obtain a decision tree $\Gamma
^{\prime }$ over $z$, which solves the problem $z$ deterministically and in
which each working node has exactly two leaving edges and each complete path
is realizable.
Let $\xi _{1}$ and $\xi _{2}$ be different complete paths
in $\Gamma ^{\prime }$, $a_{1}\in A(\xi _{1})$, and $a_{2}\in A(\xi _{2})$. It is easy
to show that $(f_{1}(a_{1}),\ldots ,f_{n}(a_{1}))\neq (f_{1}(a_{2}),\ldots
,f_{n}(a_{2}))$.
Therefore $L_{t}(\Gamma ^{\prime })\leq N_{U}(f_{1},\ldots ,f_{n})\leq N_{U}(n)$. It is clear that
$\Gamma ^{\prime } \in G_{d}^{2}(U)$.
 By Lemma \ref{L3}, $%
L_{w}(\Gamma ^{\prime })= L_{t}(\Gamma ^{\prime })-1$. Therefore $L(\Gamma ^{\prime })\leq
2N_{U}(n)$. By Proposition \ref{P5}, $2N_{U}(n)=L_{U}^{ld}(n)$. Taking into account that $h(\Gamma ^{\prime
})\leq h_{U}^{ld}(n)$ and $z$ is an arbitrary problem over $U$ with $%
\dim z=n$, we obtain that $U$ is $ld$-reachable. \qed
\end{proof}

\begin{lemma}
\label{L8}Let $U$ be an infinite binary information system such that $%
h_{U}^{la}(n)=n$ for any $n\in \mathbb{N}$. Then the information system $U$
is $la$-reachable.
\end{lemma}

\begin{proof}
Let $z=(\nu ,f_{1},\ldots ,f_{n})$ be a problem over $U$ and $\Gamma $ be a
decision tree over $z$ that solves the problem $z$ deterministically and
satisfies the following conditions: the number of working nodes in each
complete path of $\Gamma $ is equal to $n$ and these nodes in the order from
the root to a terminal node are labeled with attributes $f_{1},\ldots ,f_{n}$%
. Remove from $\Gamma $ all nodes and edges that do not belong to realizable
complete paths. Let $w$ be a working node in the obtained tree that has only
one leaving edge $d$ entering a node $v$. We remove the node $w$ and edge $%
d$ and connect the edge $e$ entering $w$ to the node $v$. We do the same
with all working nodes with only one leaving edge. Denote by $\Gamma
^{\prime }$ the obtained decision tree. It is clear that $\Gamma ^{\prime }$
solves the problem $z$ deterministically and hence nondeterministically, $%
\Gamma ^{\prime }\in G_{d}^{2}(U)$, and $L_{t}(\Gamma ^{\prime })\leq
N_{U}(f_{1},\ldots ,f_{n})\leq N_{U}(n)$. By Lemma \ref{L3}, $L_{w}(\Gamma
^{\prime })=L_{t}(\Gamma ^{\prime })-1$. Therefore $L(\Gamma ^{\prime })\leq
2N_{U}(n)$. Using Proposition \ref{P5}, we obtain $L(\Gamma ^{\prime })\leq
L_{U}^{la}(n)$. It is clear that $h(\Gamma ^{\prime })\leq n=h_{U}^{la}(n)$.
Therefore $U$ is $la$-reachable. \qed
\end{proof}

\begin{lemma}
\label{L9}Let $U$ be an infinite binary information system, which satisfies
the condition of reduction. Then the information system $U$ is not $la$%
-reachable.
\end{lemma}

\begin{proof}
By Proposition \ref{P2}, the function $h_{U}^{la}(n)$ is bounded from above
by a positive constant $c$. By Proposition \ref{P6}, the function $N_{U}(n)$
is not bounded from above by a constant. Choose $n\in \mathbb{N}$ such that $%
N_{U}(n)>2^{2c}$. Let $z=(\nu ,f_{1},\ldots ,f_{n})$ be a problem over $U$
such that $\nu (\bar{\delta}_{1})\neq \nu (\bar{\delta}_{2})$ for any $\bar{%
\delta}_{1},\bar{\delta}_{2}\in \{0,1\}^{n}$, $\bar{\delta}_{1}\neq \bar{%
\delta}_{2}$, and $N_{U}(f_{1},\ldots ,f_{n})=N_{U}(n)$. Let $\Gamma $ be a
decision tree over $z$, which solves the problem $z$ nondeterministically,
for which $h(\Gamma )\leq h_{U}^{la}(n)\leq c$, and which has the minimum
number of nodes among such trees. In the same way as it was done in the
proof of Lemma \ref{L2}, we can prove that $\Gamma \in G_{a}^{f}(U)$.
It is clear that $L_{t}(\Gamma )\geq N_{U}(f_{1},\ldots ,f_{n})=N_{U}(n)$.
Let us assume that $\Gamma \in G_{d}^{2}(U)$. Then it is easy to show that $%
h(\Gamma )\geq \log _{2}L_{t}(\Gamma )\geq \log _{2}N_{U}(n)>2c$, which is
impossible by the choice of $\Gamma $. Therefore $\Gamma \in
G_{a}^{f}(U)\setminus G_{d}^{2}(U)$. By Lemma \ref{L4}, $L_{w}(\Gamma
)>L_{t}(\Gamma )-1\geq N_{U}(n)-1$. Using Proposition \ref{P5}, we obtain $%
L(\Gamma )>2N_{U}(n)=L_{U}^{la}(n)$. Therefore $U$ is not $la$-reachable. \qed
\end{proof}

\begin{lemma}
\label{L10}Let $U=(A,F)$ be an infinite binary information system, which satisfies
the condition of reduction with parameter $m$. Then $%
(m,(m+1)L_{U}^{la}(n)/2+1)$ is a boundary $la$-pair of the system $U$.
\end{lemma}

\begin{proof}
Let $z=(\nu ,f_{1},\ldots ,f_{n})$ be a problem over $U$. We now describe a
decision tree $\Gamma $ over $z$, which solves the problem $z$
nondeterministically and for which $h(\Gamma )\leq m$ and  $L(\Gamma )\leq (m+1)L_{U}^{la}(n)/2+1$. For each tuple
$\bar{\delta}=(\delta _{1},\ldots ,\delta _{n})\in \{0,1\}^{n}$ for which the
system of equations
\begin{equation*}
S_{\bar{\delta}}=\{f_{1}(x)=\delta _{1},\ldots ,f_{n}(x)=\delta _{n}\}
\end{equation*}%
has a solution from $A$, we describe a complete path $\xi _{\bar{\delta}}$.
Since the information system $U$ satisfies the condition of reduction with
parameter $m$, there exists a subsystem
\begin{equation*}
S_{\bar{\delta}}^{\prime }=\{f_{i_{1}}(x)=\delta _{i_{1}},\ldots
,f_{i_{t}}(x)=\delta _{i_{t}}\}
\end{equation*}%
of the system $S_{\bar{\delta}}$, which has the same set of solutions and
for which $t\leq m$. Then
\begin{equation*}
\xi _{\bar{\delta}}=v_{0},d_{0},v_{1},d_{1},\ldots ,v_{t},d_{t},v_{t+1},
\end{equation*}%
where the node $v_{0}$ and the edge $d_{0}$ are not labeled, for $j=1,\ldots
,t$, the node $v_{j}$ is labeled with the attribute $f_{i_{j}}$ and the edge
$d_{j}$ is labeled with the number $\delta_{i_{j}}$, and the node $v_{t+1}$ is
labeled with the number $\nu (\bar{\delta})$. We merge initial nodes of all
such complete paths and denote by $\Gamma $ the obtained tree. One can show
that $\Gamma $ is a decision tree over $z$, which solves the problem $z$
nondeterministically and for which $h(\Gamma )\leq m$. The number of the
considered complete paths is equal to $N_{U}(f_{1},\ldots ,f_{n})\leq
N_{U}(n)$. The number of nodes in each complete paths is at most $m+2$.
Therefore $L(\Gamma )\leq (m+1)N_{U}(n)+1$. By Proposition \ref{P5}, $%
N_{U}(n)=L_{U}^{la}(n)/2$. Hence $L(\Gamma )\leq (m+1)L_{U}^{la}(n)/2+1$.
Thus, $(m,(m+1)L_{U}^{la}(n)/2+1)$ is a boundary $la$-pair of the system $U$. \qed
\end{proof}

\begin{proof}[Proof of Theorem \protect\ref{T2}]
Each information system from the class $W_{1}^l$ satisfies the condition of
reduction (see Table \ref{tab2}).

(a) Let $U$ be an information system from the class $W_{1}^l$. Using Lemma \ref%
{L7}, we obtain that the system $U$ is $ld$-reachable.

(b) Let $U$ be an information system from the class $W_{1}^l$. Then, for
some $m\in \mathbb{N}$, the system $U$ satisfies the condition of reduction
with parameter $m$. Using Lemma \ref{L9}, we obtain that the system $U$ is
not $la$-reachable. Using Lemma \ref{L10}, we obtain that $%
(m,(m+1)L_{U}^{la}(n)/2+1)$ is a boundary $la$-pair of the system $U$. \qed
\end{proof}

\begin{proof}[Proof of Theorem \protect\ref{T3}]
Each information system from the class $W_{2}^l$ does not satisfy the
condition of reduction (see Table \ref{tab2}).

(a) Let $U$ be an information system from the class $W_{2}^l$. Using Lemma \ref%
{L7}, we obtain that the system $U$ is $ld$-reachable.

(b) Let $U$ be an information system from the class $W_{2}^l$. By Proposition %
\ref{P2}, $h_{U}^{la}(n)=n$ for any $n\in \mathbb{N}$. Using Lemma \ref{L8},
we obtain that the system $U$ is $la$-reachable. \qed
\end{proof}

\begin{proof}[Proof of Theorem \protect\ref{T4}]
Each information system from the class $W_{3}^l$ does not satisfy the
condition of reduction (see
Table \ref{tab2}).

(a) Let $U$ be an information system from the class $W_{3}^l$. Using Lemma \ref%
{L7}, we obtain that the system $U$ is $ld$-reachable.

(b) Let $U$ be an information system from the class $W_{3}^l$.  By
Proposition \ref{P2}, $h_{U}^{la}(n)=n$ for any $n\in \mathbb{N}$. Using
Lemma \ref{L8}, we obtain that the system $U$ is $la$-reachable. \qed
\end{proof}

\section{Conclusions} \label{S6}

In this paper, we divided the set of all infinite binary information systems  into three complexity classes depending on the worst case time and space complexity of deterministic and nondeterministic decision trees. This allowed us to identify nontrivial relationships between deterministic decision trees and decision rule systems represented by nondeterministic decision trees. For each complexity class, we studied issues related to time-space trade-off for deterministic and nondeterministic decision trees. In the future, we are planning to generalize the obtained results to the case of classes of decision tables closed under operations of removal of attributes and changing decisions attached to rows of decision tables.

\subsection*{Acknowledgements}

Research reported in this publication was
supported by King Abdullah University of Science and Technology (KAUST).

\bibliographystyle{spmpsci}
\bibliography{1C}

\end{document}